\documentclass[pra,twocolumn,floatfix,superscriptaddress,longbibliography,notitlepage,nofootinbib]{revtex4-1}

\usepackage{graphicx, color, graphpap}
\usepackage{float}
\usepackage{enumitem}
\usepackage{physics}
\usepackage{amssymb}
\usepackage{amsthm}
\usepackage{multirow}
\usepackage[colorlinks=true,citecolor=blue,linkcolor=magenta]{hyperref}
\usepackage[T1]{fontenc}
\usepackage{bbm}
\usepackage{thmtools,thm-restate}
\usepackage{verbatim}
\usepackage{mathtools}
\usepackage{titlesec}
\usepackage{amsmath}
\usepackage[title]{appendix}

\usepackage[caption = false]{subfig}

\long\def\ca#1\cb{} 

\newcommand{\EC}{\mathcal{E}}

\newcommand{\SC}{\mathcal{S}}

\renewcommand{\geq}{\geqslant}
\renewcommand{\leq}{\leqslant}

\DeclareMathOperator*{\argmax}{arg\,max}
\DeclareMathOperator*{\argmin}{arg\,min}
\renewcommand{\vec}[1]{\boldsymbol{#1}}  

\newcommand{\ad}{^\dagger}

{}
{}
\newtheorem{theorem}{Theorem}
\newtheorem{lemma}{Lemma}

\begin{document}

\title{The quantum low-rank approximation problem}

\author{Nic Ezzell}
\affiliation{Department of Physics \& Astronomy, University of Southern California, Los Angeles, California 90089, USA}
\affiliation{Information Sciences, Los Alamos National Laboratory, Los Alamos, NM 87545, USA}
\author{Zo\"{e} Holmes}
\affiliation{Information Sciences, Los Alamos National Laboratory, Los Alamos, NM 87545, USA}
\affiliation{Centre for Nonlinear Studies, Los Alamos National Laboratory, Los Alamos, NM 87545, USA}

\author{Patrick J. Coles}
\affiliation{Theoretical Division, Los Alamos National Laboratory, Los Alamos, NM 87545, USA}

\begin{abstract}
We consider a quantum version of the famous low-rank approximation problem. Specifically, we consider the distance $D(\rho,\sigma)$ between two normalized quantum states, $\rho$ and $\sigma$, where the rank of $\sigma$ is constrained to be at most $R$. For both the trace distance and Hilbert-Schmidt distance, we analytically solve for the optimal state $\sigma$ that minimizes this distance. For the Hilbert-Schmidt distance, the unique optimal state is $\sigma = \tau_R +N_R$, where $\tau_R = \Pi_R \rho \Pi_R$ is given by projecting $\rho$ onto its $R$ principal components with projector $\Pi_R$, and $N_R$ is a normalization factor given by $N_R = \frac{1- \Tr(\tau_R)}{R}\Pi_R$. For the trace distance, this state is also optimal but not uniquely optimal, and we provide the full set of states that are optimal. We briefly discuss how our results have application for performing principal component analysis (PCA) via variational optimization on quantum computers.
\end{abstract}

\maketitle

\section{Introduction}

Approximating a given matrix with a low-rank matrix has many applications in mathematics, engineering, and data science, such as for natural language processing or for linear systems identification~\cite{LRA_wikipedia}. This task is known as the low-rank approximation problem~\cite{eckart_approximation_1936, jolliffe_principal_2016, johnson_matrix_1990, markovsky_structured_2008, markovsky_low_2012}. The generic structure of the problem is as follows: For a given matrix $A$, find the matrix $B$ that minimizes some distance measure $D(A,B)$ subject to the constraint that $\rank(B)\leq R$, and also possibly subject to additional constraints on $B$.

In the most basic formulation, there are no additional constraints on $B$ and the distance measure is typically the Frobenius (Hilbert-Schmidt) norm. In that case, the solution is for $B$ to be a truncated version of the singular value decomposition of $A$, where the truncation is needed to satisfy the rank constraint on $B$. There are also more complicated formulations involving additional constraints on $B$, such as non-negativity or structural constraints.

In this work, we consider a quantum version of the low-rank approximation problem. Naturally, this means that we will consider two quantum states $\rho$ and $\sigma$.  As these are density matrices, they come with the constraints that they are positive semi-definite and normalized (trace one). These additional constraints make the solution non-trivial, and in particular, the solution derivation does not follow directly from prior work on low-rank approximation. Hence, the ``quantum low-rank approximation problem'' appears to be a novel problem.
We tackle this problem for two alternative distance measures: the Hilbert-Schmidt distance and the trace distance~\cite{nielsen2000quantum,wilde_classical_2017}. We note that the problem of finding the state that minimizes a relevant Hilbert-Schmidt distance (potentially subject to constraints) is investigated in other contexts in Refs.~\cite{Smolin2012Efficient, Ben2018What}.

The trace distance is often viewed as a more useful measure than the Hilbert-Schmidt distance, due to its properties and operational meaning. However, in the context of the quantum low-rank approximation problem, we argue that the opposite is true. We show that Hilbert-Schmidt distance has a unique solution to the quantum low-rank approximation problem, whereas the trace distance has a highly degenerate family of solutions. Interestingly, the same state is optimal for both distance measures, but that state is not unique in the trace distance case. As a consequence of this degeneracy issue, we argue that the Hilbert-Schmidt distance is a more useful measure for applications of the quantum low-rank approximation problem.

The most important application of this problem is principal component analysis (PCA). Specifically, it has been proposed to use quantum computers to perform PCA~\cite{lloyd2014quantum}, potentially with a significant speedup over classical computation. While the original quantum PCA algorithm involves deep circuits and many qubit resources, there has been recent effort to making quantum PCA implementable in the near term through variational approaches~\cite{larose2019variational,cerezo2020variational,verdon2019quantum,cerezo2020variationalreview}. Our work has implications along these lines. Namely, our analytical results imply that one could variationally optimize over low-rank states $\sigma$ ($\rank(\sigma)\leq R$) in order to minimize the Hilbert-Schmidt distance $D_{HS}(\rho,\sigma)$. The resulting optimized state $\sigma^* = \arg \min D_{HS}(\rho,\sigma)$ would essentially be a truncation of $\rho$ up to the $R$th principal components. Hence, this would be a means to extract the principal components of $\rho$ into the optimized state $\sigma^*$. We elaborate on the relevance of our results to PCA later in this article.

This problem also has applications for the compressibility of quantum states. Typically, the resources required to prepare, store or simulate a quantum state on a quantum computer increase with increasing rank. For example, if you store a quantum state via its purification, the number of ancilla qubits required to prepare a state of rank $r$ is $\log(r)$. Thus it may be desirable to learn how to prepare a low rank approximation to a given state in order to have a compressed approximate representation of such state. Our solution to the quantum low-rank approximation problem gives the best low rank (i.e., compressed) state to well approximate a given target state.

\section{Background: Low-rank approximation problem}

The generic low rank-approximation problem amounts to finding the best approximation of a matrix $A$ using a generally lower rank matrix $B$~\cite{markovsky_low_2012}. More formally, if $\rank(B) \leq R \leq \rank(A) = r$, we solve for $B$ that minimizes some distance measure $D(A, B)$ possibly subject to additional constraints. When $D$ is the Frobenius distance, $D(A, B) = ||A - B||_F$ and no other constraints beyond the rank constraint are applied, the problem has a unique solution given by the Eckart–Young–Mirsky theorem~\cite{eckart_approximation_1936}. The solution is most easily stated by invoking the spectral decomposition theorem which states that any matrix $A$ can be written as $A = U \Lambda V^\dagger$ for $U$, $V$ unitary matrices and $\Lambda$ a diagonal matrix with positive, ordered entries, $\lambda_1 \geq \lambda_2, \ldots, \lambda_r$ known as the singular values. The unique optimal solution is then simply $B^*(R) = \sum_{i=1}^R \lambda_i u_i v^\dagger_i$, where $\{u_i\}$ and $\{v_i\}$ are the eigenvectors of $U$ and $V$ respectively.\footnote{The same solution also works for a variety of other matrix norms such as the spectral norm $D(A, B) = ||A - B||_2$~\cite{eckart_approximation_1936}.}

Several comments are in order. First and unsurprisingly, the optimum occurs when $\rank\left(B^*(R)\right) = R$. Second, the approximation error is characterized by the highest $R$ singular values of $A$. In particular, $||A - B^*(R)||_F = \sum_{i=R+1}^{r}\lambda_i^2$, so indeed, the quality of the approximation depends only on $R$ relative to $r$. Finally, if $A$ is a square, positive semi-definite matrix then this problem is equivalent to finding the principal components of $A$~\cite{jolliffe_principal_2016}. In particular, we can interpret $A$ as a data covariance matrix $A = X^T X$ so that $\sqrt{\lambda_i}$ for $i = 1, \ldots, R$ are the principal components of $X$. In addition, $U = V$ in this case, so the optimal matrix $B$ is such that (i) it is diagonal in the same basis as $A$, and (ii) it contains the $R$ highest eigenvalues or $A$ (i.e. the principal components).

Variations on this general result come from imposing different constraints on the approximation $B$. Some common ones include positive (semi-)definiteness as in PCA~\cite{jolliffe_principal_2016} and matrix completion~\cite{markovsky_low_2012}, missing entries as in recommended systems~\cite{johnson_matrix_1990}, Hankel structure in system identification~\cite{markovsky_structured_2008}, and so on.

\section{Quantum low-rank approximation problem}

\subsection{General set-up}

In the quantum case, we replace $A$ and $B$ with two (finite-dimensional) quantum states $\rho$ and $\sigma$ where $\rho$ is the target state. On physical grounds, we require $\rho$ and $\sigma$ to be Hermitian (i.e. positive semi-definite with complex entries) and normalized $\Tr[\rho] = \Tr[\sigma] = 1$~\cite{nielsen2000quantum}. Other than that, the problem remains the same: find the optimal rank-constrained state,
\begin{equation}
    \label{eq:quantum-low-rank-approx-prob}
    \sigma^*(R) = \argmin_{\sigma \geq 0,\rank(\sigma) \leq R, \Tr(\sigma)=1} D(\rho, \sigma),
\end{equation}
for some distance measure $D$. We call this the \textit{quantum low-rank approximation problem} (QLRAP). In this work, we solve the QLRAP for the Hilbert-Schmidt distance $D_{HS}$ and the trace-norm distance $D_T$.

Note that $D_{HS} = D_{F}$ in finite dimensions, so if we drop the trace constraint on $\sigma$, then this problem reduces to finding the principal components of $\rho$.
To clarify this point, let us first introduce some useful notation we shall use from here on out. Since $\rho \geq 0$, we may decompose it as
\begin{equation}
    \rho = \sum_{i=1}^r \lambda_i \ketbra{e_i}
\end{equation}
for $\lambda_i > 0$, $\sum_{i=1}^r \lambda_i = 1$,  $\rank(\rho) = r$, and some orthonormal basis $\{\ket{e_i}\}$. Further, we shall assume for simplicity that the eigenvalues or $\rho$ are ordered and non-degenerate, $\lambda_1 > \lambda_2 > \ldots > \lambda_r$. By the Eckart-Young-Mirsky theorem, 
\begin{equation}
    \tau^*(R) = \argmin_{\tau \geq 0,\rank(\tau)\leq R} D_{HS}(\rho, \tau)
\end{equation}
has solution
\begin{equation}
    \tau^*(R) = \Pi_R \rho \Pi_R \equiv \tau_R, \ \ \ \Pi_R = \sum_{i=1}^R \ketbra{e_i},
\end{equation}
where $\Pi_R$ is simply a projector onto the $R$ principal components of $\rho$. The resulting matrix is Hermitian but clearly only normalized when $R = r$, so it cannot generally be a valid quantum state. In the remaining text, we solve Eq.~\eqref{eq:quantum-low-rank-approx-prob} for $D = D_{HS}$ and $D = D_{T}$ and then compare their usefulness in the context of quantum PCA. 

Note that we assumed the eigenvalues of $\rho$ are ordered and non-degenerate so that $\Pi_R$ has a unique definition. If we drop this assumption, then $\Pi_R$ is arbitrary up to a permutation of the indices which doesn't fundamentally change any of our results but makes the notation more cumbersome. Hence, we shall continue to make this assumption through all our proofs.

\subsection{Hilbert-Schmidt Distance}

We will now consider the QLRAP when the distance measure is the Hilbert-Schmidt distance. The Hilbert-Schmidt distance is given by~\cite{dodonov_hilbert-schmidt_2000, ozawa_entanglement_2000, wilde_classical_2017}
\begin{equation}
    D_{HS}(\rho, \sigma) = \Tr [(\rho - \sigma)^2].
\end{equation}
In what follows, we will prove that the optimal state, in this case, has the form:
\begin{equation}
\label{eqn_optimalstate_dhs}
    \sigma^*(R) = \tau_R + \left(\frac{1-\Tr[\tau_R]}{R}\right)\Pi_R\,.
\end{equation}
Interestingly, one can think of this state as corresponding to the principal components of $\rho$ (given by $\tau_R$) plus an additive normalization factor. Naively, one might guess that the normalization factor would be multiplicative and that the optimal state would simply be $\tau_R / \Tr[\tau_R]$. However, this naive guess is wrong, and there exist numerous counterexamples demonstrating that $\tau_R / \Tr[\tau_R]$ is not the optimal state (for one example see Fig.~\ref{fig:contour-plots}a). Instead, we argue that the normalization must be additive, not multiplicative, in order for the state to optimal.

We break up our proof of the optimal $\sigma^*(R)$ into two steps. The first step is the following lemma, where we show that $\rho$ and $\sigma^*(R)$ must be diagonal in the same basis. The proof involves Schur convexity and majorization.

\begin{lemma}\label{lemma1}
For the QLRAP with $D = D_{HS}$, the optimal state $\sigma^*(R)$ must be diagonal in the same basis as $\rho$. 
\end{lemma}
\begin{proof}
Consider a family of states related by unitary rotation: $\sigma_U = U \sigma U\ad$. Note that all states in this family have the same purity: $\Tr(\sigma_U^2) = \Tr(\sigma^2)$. Hence, the purity terms in the Hilbert-Schmidt distance, $\Tr(\rho^2) + \Tr(\sigma_U^2) - 2 \Tr(\rho \sigma_U)$, are fixed quantities for this family. Therefore, one can simply focus on the overlap term, $\Tr(\rho \sigma_U)$, to understand which state $\sigma_U$ in the family minimizes the Hilbert-Schmidt distance. Note that one can write 
\begin{equation}
    \Tr(\rho \sigma_U) = \vec{\lambda}\cdot \vec{s}_U 
\end{equation}
where $\vec{\lambda}$ is the vector of eigenvalues of $\rho$ listed in decreasing order and $\vec{s}_U$ are the diagonal elements of $\sigma_U$ in the eigenbasis of $\rho$, i.e., ${\vec{s}_U}_i = \bra{e_i}\sigma_U\ket{e_i}.$

Given a vector $\vec{v}$ in descending order, the dot product $ \vec{v} \cdot \vec{x}$ is a Schur convex function of $\vec{x}$. Hence, $\Tr[\rho \sigma_U]$ is a Schur convex function of $\vec{s}_U$. Next we note that the eigenvalues of a positive semi-definite matrix always majorize the diagonal elements in any basis, i.e., $\vec{\lambda}_{\sigma} \succ \vec{s}_U$. Here, $\vec{\lambda}_{\sigma}$ are the eigenvalues of $\sigma$ (in decreasing order), and hence these are also the eigenvalues of every state in the family $\sigma_U$. Due to this majorization relation and the Schur convexity property, we have that
\begin{align}
 \Tr(\rho \sigma_U) &= \vec{\lambda}\cdot \vec{s}_U \\
 &\leq \vec{\lambda}\cdot \vec{\lambda}_{\sigma}\\
 & = \vec{\lambda}\cdot \vec{s}_V \\ 
 & = \Tr(\rho \sigma_{V}) \, ,
\end{align}
where $\sigma_{V} = V \sigma V\ad$, ${\vec{s}_V}_i = \bra{e_i}\sigma_V\ket{e_i}$ and $V$ satisfies $V = \argmax_U \Tr(\rho \sigma_{U})$. Since $\vec{s}_V = \vec{\lambda}_{\sigma}$ we have that $\sigma_{V}$ is diagonal in the same basis as $\rho$. 
Hence, the state that maximizes $\Tr(\rho \sigma_{U})$, and therefore minimizes the Hilbert-Schmidt distance, is one that is diagonal in the same basis as $\rho$.
\end{proof}

With this lemma in hand, we now provide the proof that the optimal $\sigma^*(R)$ for minimizing $D_{HS}$ is given by the formula in \eqref{eqn_optimalstate_dhs}. This completely characterizes the quantum low-rank approximation solution for $D = D_{HS}$. We offer two alternative proofs, one based on the  Lagrangian dual problem, and one based on Schur convexity and majorization.

\begin{theorem}\label{theorem1}
For the QLRAP with $D=D_{HS}$, the optimal state $\sigma^*(R)$ is unique and is given by
\begin{subequations}
    \begin{align}
    \sigma^*(R) &= \tau_R + N_R \\ 
    \tau_R &\equiv \Pi_R \rho \Pi_R \\
    N_R &\equiv \frac{1 - \Tr[\tau_R]}{R} \Pi_R
    \end{align}
\end{subequations}
\end{theorem}
\begin{proof}

\textbf{Proof 1:} Our first proof is based on the method of Lagrange multipliers~\cite{protter_basic_1998}. By Lemma~\ref{lemma1}, we can rewrite the Hilbert-Schmidt distance as a quadratic form,
\begin{align}
    D_{HS}(\rho, \sigma) &= ||\vec{\lambda}||_2^2 + ||\vec{\lambda_{\sigma}}||_2^2 - 2 \vec{\lambda} \cdot \vec{\lambda_{\sigma}}\,.
\end{align}
Here, $\vec{\lambda} = (\lambda_1, \lambda_2, \ldots, \lambda_r)$ are the fixed eigenvalues of $\rho$ ordered in decreasing order, and $\vec{\lambda_{\sigma}} = (\lambda_{\sigma 1}, \lambda_{\sigma 2}, \ldots, \lambda_{\sigma R}, 0, \ldots, 0\}$ are the $R$ variable eigenvalues (not necessarily ordered) of $\sigma^*(R)$  appended with $r - R$ zeros to match dimension. Note that $\lambda_{\sigma i} \geq 0$, so we are not excluding the possibility that the solution has rank less than $R$. By using Lagrange multipliers $\alpha, \beta$, we may include the constraints $\rank(\sigma) = R$, $\Tr[\sigma] = 1$, and $\Tr[\rho] = 1$ into a Lagrangian, 
\begin{multline}
    L = \sum_{i=1}^R (\lambda_{\sigma i} - \lambda_i)^2 + \sum_{i=R+1}^r \lambda_i^2\\
    -\alpha (1 - \sum_{k=1}^R \lambda_{\sigma k}) - \beta (1 - \sum_{k=1}^r \lambda_k).
\end{multline}
Taking first derivatives, we get
\begin{align}
    \label{eq:lambda-div}
    \pdv{C}{\lambda_{\sigma k}} &= 2 (\lambda_{\sigma k} - \lambda_k) - \alpha = 0 \\
    \label{eq:alpha-div}
    \pdv{C}{\alpha}&= 1 - \sum_{k=1}^R \lambda_{\sigma k} = 0 \\
    \label{eq:beta-div}
    \pdv{C}{\beta} &= 1 - \sum_{k=1}^r \lambda_k = 0. 
\end{align}
Plugging Eqs.~\eqref{eq:alpha-div} and \eqref{eq:beta-div} into Eq.~\eqref{eq:lambda-div} and summing over the first $R$ entries, we get
\begin{subequations}
    \begin{align}
        2 \sum_{k=1}^R \lambda_{\sigma k} - 2 \sum_{k=1}^R \lambda_k &= R \alpha \\
        2 \left( 1 - \Tr[\Pi_R \vec{\lambda}] \right) &= R \alpha \\
        \therefore \alpha^* &= \frac{2 (1 - \Tr[\Pi_R \vec{\lambda}])}{R}.
    \end{align}
\end{subequations}
Plugging this $\alpha^*$ back into Eq.~\eqref{eq:lambda-div}, we get
\begin{equation}
    \lambda_{\sigma k}^* = \lambda_k + \frac{(1 - \Tr[\Pi_R \vec{\lambda}])}{R} \,,\quad \forall k \leq R,
\end{equation}
and $\lambda_{\sigma k}^* = 0$ for $k  > R$ is already encoded into the form of $\vec{\lambda_{\sigma}}$. Note that $\lambda_{\sigma i} \geq \lambda_{\sigma j}$ for $i > j$ arises naturally from the assumed eigenvalue ordering of $\vec{\lambda}$, but was not an assumption of this proof. Turning these vectors into diagonal matrices, we get the desired result. Because we began with a quadratic form which is strongly convex, the found solution is not only a global minimum, but also a unique global minimum.

\textbf{Proof 2:} Our second proof is based on Schur convexity. We write the Hilbert-Schmidt distance as
\begin{align}
    D_{HS}(\rho,\sigma) &= \sum_{i} (\lambda_{\sigma i} - \lambda_{ i})^2 \label{eqn_proof2_1}\\
    &= \sum_{i\in \SC} (\lambda_{\sigma i} - \lambda_{ i})^2 + \sum_{i\notin \SC} \lambda_{ i}^2
\end{align}
where $\SC$ is the set of indices over which $\sigma$ has support. Consider the vector 
\begin{equation}
 \vec{q} = \{q_i\}_{i\in \SC} = \{ \abs{\lambda_{\sigma i} - \lambda_{ i}}\}_{i\in \SC},
\end{equation}
where we assume that the elements $q_i$ are ordered in decreasing order. Note that we can write
\begin{align}
    D_{HS}(\rho,\sigma) = \sum_{i\in \SC} q_i^2 + \sum_{i\notin \SC} \lambda_{ i}^2
\end{align}
Also, consider a vector (of size $|\SC |$) with uniform elements $\vec{t}= \{Q/|\SC |\}$, where $Q = \sum_{i\in \SC} q_i$ and $|\SC |$ is the cardinality of $\SC$. Then it is clear that we have the majorization relation $\vec{q}\succ \vec{t}$, since $\vec{t}$ is a completely flat vector and $\sum_i q_i = \sum_i t_i $. In addition, for another uniform vector $\vec{t'} = \{t_i'\}$ of size $|\SC|$, we have $\vec{t}\succ \vec{t'}$, where $t_i' = {Q'/|\SC |}$ for all $i$. Here we have defined 
\begin{equation}
Q' = \sum_{i\in \SC} (\lambda_{\sigma i} - \lambda_{ i}) = 1 -  \sum_{i\in \SC} \lambda_{ i} = \sum_{i\notin \SC} \lambda_{ i} \,.    
\end{equation}
This follows simply because $Q \geq Q'$. Hence we have that
\begin{equation}
\vec{q}\succ \vec{t} \succ \vec{t'}\,.  \end{equation}
Now note that the function $f(\vec{x}) = \sum_i x_i^2$ is a (strictly) Schur convex function of $\vec{x}$~\cite{aniello2016characterizing}. This implies that $f(\vec{q})\geq f(\vec{t'})$. Therefore we have
\begin{align}
    D_{HS}(\rho,\sigma) &\geq \sum_{i\in \SC} (t'_i)^2 + \sum_{i\notin \SC} \lambda_{ i}^2\\
    & = \sum_{i\in \SC} (Q'/|\SC |)^2 + \sum_{i\notin \SC} \lambda_{ i}^2\,.
\end{align}
This lower bound on the Hilbert-Schmidt distance is achievable if we set the eigenvalues of $\sigma$ to be such that $\lambda_{\sigma i} = \lambda_i + Q'/|S|$. Hence, for a fixed choice of $\SC$, the minimal value of $D_{HS}(\rho,\sigma)$ is given by:
\begin{equation}
\begin{aligned}
    g(\SC) &= \sum_{i\in \SC} (Q'/|\SC |)^2 + \sum_{i\notin \SC} \lambda_{ i}^2\, \\
    &= \frac{1}{|\mathcal{S}|} \left(\sum_{i\notin \SC} \lambda_{ i}\right)^2+  \sum_{i\notin \SC} \lambda_{ i}^2\,.\label{eqn_gofS_proof2}
\end{aligned}
\end{equation} 

Let us consider minimizing the expression in~\eqref{eqn_gofS_proof2} over all $\SC$, subject to the constraint that $|\SC | \leq R$. First, it is clear that the expression is minimized when $|\SC|$ is as large as possible, and hence when $|\SC| = R$. Moreover, the terms $\sum_{i\notin \SC} \lambda_{ i}$ and $\sum_{i\notin \SC} \lambda_{ i}^2$ are minimized whenever $\SC$ contains the $R$ largest eigenvalues of $\rho$, and hence the complement of $\SC$ contains the $r - R$ smallest eigenvalues of $\rho$. Therefore, all of support of $\sigma$ should be concentrated on the $R$ largest eigenvalues of $\rho$, in order to be optimal. Hence the optimal eigenvalues of $\sigma$ are $\lambda_{\sigma i} = \lambda_i + (1-\sum_{i=1}^R \lambda_i)/R$ for $i\leq R$ and $\lambda_{\sigma i} = 0$ for $i >R$, assuming the index $i$ is ordered according to the eigenvalue ordering of $\rho$. Note that this spectrum corresponds precisely to the spectrum of $\tau_R + N_R$. Hence, combined with Lemma~\ref{lemma1}, we have completed the proof that $\tau_R + N_R$ is the optimal state. Finally, noting that the function $f(\vec{x}) = \sum_i x_i^2$ is strictly Schur convex~\cite{aniello2016characterizing} implies that any other choice of spectrum would increase the value of the objective function, and hence the state $\tau_R + N_R$ is uniquely optimal.
\end{proof}

Amazingly, the only difference between $\sigma^*(R)$ with a trace constraint and $\tau^*(R)$ without it is a simple additive normalization term $N_R$. In both cases, the solution is unique which we explore visually in Fig.~\ref{fig:contour-plots}(a). Finally, given the optimal low-rank approximation, the corresponding minimum Hilbert-Schmidt distance is
\begin{subequations}
    \begin{align}
         D^*_{HS}(\rho, \sigma^*(R)) &= \Tr[(I - \Pi_R) \rho^2] + \Tr[N_R^2] \\
         &= \sum_{i=R + 1}^r \lambda_i^2 + \sum_{i=1}^R (\lambda_{\sigma i} - \lambda_i)^2,
    \end{align}
\end{subequations}
where the second equality gives an intuitive interpretation. The first term accounts for the $r - R$ eigenvalues we truncate in our approximation. The second term accounts for the constant, additive error between the original eigenvalues $\lambda_i$ and their corresponding approximated values $\lambda_{\sigma i}$. Indeed, we can begin with this intuitive optimal distance and then show that $\sum_{i=R + 1}^r \lambda_i^2 = \Tr[(I - \Pi_R) \rho^2] $ and $  \sum_{i=1}^R (\lambda_{\sigma i} - \lambda_i)^2 = \Tr[N_R^2].$

From this interpretation, $D^*_{HS}$ is clearly a monotonically decreasing function with increasing $R$ that reaches $D^*_{HS} = 0$ when $R \geq r$ but no sooner. This captures the intuitive and obvious notion that the rank $R$ approximation gets better with larger $R$.   

We note that the low rank approximation problem for the Hilbert-Schmidt distance, Eq.~\eqref{eqn_optimalstate_dhs}, bares structural similarities to the problem of finding the state that minimizes the distance from an arbitrary normalized Hermitian operator~\cite{Smolin2012Efficient}. However, conceptually these two problems are rather different in that the latter involves no rank constraint and is not looking at the distance between two states.

\subsection{Trace distance}
The Trace distance is given by~\cite{nielsen2000quantum, wilde_classical_2017}
\begin{equation}
    D_{T}(\rho, \sigma) = \frac{1}{2}\Tr [|\rho - \sigma|].
\end{equation}
We break up our proof of the optimal $\sigma^*(R)$ for the trace norm into several steps, with the following lemma being the first step.

\begin{lemma}\label{lemma2}
Let $\rho$ and $\sigma$ be arbitrary quantum states. Let $\tilde{\Pi}$ be a projector that is diagonal in the same basis as $\rho$. Then there exists another state $\sigma' = U\sigma U\ad$ that is a unitarily related to $\sigma$ and is diagonal in the same basis as $\rho$, with the property that
\begin{equation}
    D_T(\rho,\sigma) \geq \Tr [\tilde{\Pi} (\rho - \sigma')]\,.
\end{equation}
\end{lemma}
\begin{proof}
From the operational meaning of the trace distance~\cite{nielsen2000quantum}, we can write
\begin{equation}
    D_T(\rho,\sigma) = \max_{\Pi} \Tr [\Pi (\rho - \sigma)]
\end{equation}
where the optimization is over all projectors $\Pi$. Since $\tilde{\Pi}$ is a particular projector, we have
\begin{align}
    D_T(\rho,\sigma) & \geq  \Tr [\tilde{\Pi} (\rho - \sigma)] \\
&=\Tr [\tilde{\Pi} \rho] - \Tr[\tilde{\Pi}  \sigma]\\
&=\Tr [\tilde{\Pi} \rho] - f(\vec{d_{\sigma}})\label{eqn_lemma2_middlestep}
\end{align}
where $f(\vec{d_{\sigma}}) = \vec{\lambda_{\tilde{\Pi}}}\cdot \vec{d_{\sigma}}$ is a dot product of two vectors. In particular,  $\vec{\lambda_{\tilde{\Pi}}}$ is the vector of eigenvalues of $\tilde{\Pi}$ listed in decreasing order and $\vec{d_{\sigma}}$ is the vector of diagonal elements of $\sigma$ in the eigenbasis of $\rho$ (and hence also the eigenbasis of $\tilde{\Pi}$).

Now, let $\sigma' = U \sigma U\ad$ be such that its eigenvalues $\vec{\lambda_{\sigma'}}$ (listed in decreasing order) correspond to its diagonal elements $\vec{d_{\sigma'}}$ in the eigenbasis of $\rho$. So we have $\vec{\lambda_{\sigma'}} = \vec{d_{\sigma'}}$. Then it follows that
\begin{equation}
    \vec{\lambda_{\sigma'}} = \vec{d_{\sigma'}} \succ \vec{d_{\sigma}}\,,
\end{equation}
since the eigenvalues of a positive semi-definite matrix always majorize its diagonal elements in any basis. Next, we invoke the fact that $f$ is a Schur convex function, which implies that 
\begin{equation}
   f( \vec{d_{\sigma'}}) \geq f( \vec{d_{\sigma}})\,.
\end{equation}
In other words, we have $\Tr[\tilde{\Pi}  \sigma'] \geq \Tr[\tilde{\Pi}  \sigma]$. Plugging this into Eq.~\eqref{eqn_lemma2_middlestep} gives
\begin{align}
    D_T(\rho,\sigma) & \geq \Tr [\tilde{\Pi} \rho] - \Tr[\tilde{\Pi}  \sigma']\\
    &= \Tr [\tilde{\Pi} (\rho  - \sigma')]\,.
\end{align}
This is the desired inequality that we wished to prove, hence completing the proof.
\end{proof}

Next we use the above lemma to argue that the optimal state is diagonal in the same basis as $\rho$.

\begin{lemma}\label{lemma3}
For the QLRAP with $D = D_{T}$, the optimal state $\sigma^*(R)$ must be diagonal in the same basis as $\rho$. 
\end{lemma}
\begin{proof}
Let $\sigma$ be an arbitrary state. Then we will argue that there exists another state $\sigma' = U\sigma U\ad$ that is a unitarily related to $\sigma$ and is diagonal in the same basis as $\rho$, with the property that 
\begin{equation}\label{eqn_lemma3_A}
    D_T(\rho,\sigma) \geq D_T(\rho, \sigma')\,.
\end{equation}
This will allow us to argue that, for any given state $\sigma$, there is always a corresponding state $\sigma'$ that is diagonal in the same basis as $\rho$ and that is closer to $\rho$ than $\sigma$ is. Hence the optimal state is diagonal in the same basis as~$\rho$.

To make this argument, we will apply the Lemma~\ref{lemma2} with a judicious choice of $\tilde{\Pi}$. Specifically, let us choose $\tilde{\Pi} = \tilde{\Pi}_{\text{opt}}$ to be the optimal projector for distinguishing $\rho$ from $\sigma'$. 
In that case, we have:
\begin{align}
  \Tr [\tilde{\Pi}_{\text{opt}} (\rho - \sigma')] &= \max_{\Pi} \Tr [\Pi (\rho - \sigma')] \\
  &= D_T(\rho, \sigma')\,.\label{eqn_lemma3_D}
\end{align}

It is crucial to note that we can choose this choice of $\tilde{\Pi}$ and still apply Lemma~\ref{lemma2}, because the optimal projector $\tilde{\Pi}_{\text{opt}}$ is diagonal in the same basis as $\rho$. In more detail, let us define the operator $\Delta = \rho - \sigma'$. Because we have assumed that $\sigma'$ is diagonal in the same basis as $\rho$, then $\Delta$ must also be diagonal in the same basis as $\rho$. Next, we use the fact that the optimal projector $\tilde{\Pi}_{\text{opt}}$ is the one that projects onto the positive portion of the spectrum of $\Delta$~\cite{nielsen2000quantum, wilde_classical_2017}. In other words, if we expand $\Delta = Q - S$ where $Q\geq 0$ and $S\geq 0$ and $Q$ and $S$ are orthogonal, then the optimal projector projects onto the support of $Q$. Of course, both $Q$ and $S$ must be diagonal in the same basis of $\rho$, in order for $\Delta$ to satisfy this property. Therefore, the optimal projector $\tilde{\Pi}_{\text{opt}}$ must be diagonal in the same basis as $Q$ and $\Delta$, and hence diagonal in the same basis as~$\rho$.

Therefore, we can apply Lemma~\ref{lemma2} and combine it with Eq.~\eqref{eqn_lemma3_D} to obtain:
\begin{align}
  D_T(\rho,\sigma) &\geq \Tr [\tilde{\Pi}_{\text{opt}} (\rho - \sigma')]\\
  &= D_T(\rho, \sigma')\,.
\end{align}
This proves Eq.~\eqref{eqn_lemma3_A}, and hence proves the lemma.
\end{proof}

Finally, we solve for the explicit form for the set of optimal states in the following theorem. This theorem demonstrates that the state $\tau_R+N_R$ is optimal but not uniquely optimal. We further elaborate on this lack of uniqueness with an example in Sec.~\ref{sec:comp-hs-tr}.
\begin{theorem}
\label{theorem2}
For the QLRAP with $D=D_T$, one optimal solution is given by $\sigma^*(R) = \tau_R + N_R$. More generally, the set of optimal solutions is given by
\begin{align}\label{eqn_thm2_optimalset}
    \sigma^*(R) =& \bigg\{\sigma : \sigma = \sum_{i=1}^R \lambda_{\sigma i}\ketbra{e_i}, \notag\\
    &\lambda_{\sigma i} \geq \lambda_i\hspace{4pt} \text{for }i=1,...,R\bigg\}
\end{align}
where $\{ \lambda_i \}$ are the eigenvalues of $\rho$ in descending order and $\{\ket{e_i}\}_{i=1}^R$ are the eigenvectors of $\rho$ with the $R$-largest eigenvalues.
\end{theorem}
\begin{proof}
The proof that $\sigma^*(R) = \tau_R + N_R$ is an optimal solution follows a very similar path of the proof of Theorem~\ref{theorem1}. For example, one can essentially take Proof 2 of Theorem~\ref{theorem1}, replace $D_{HS}$ with $D_T$ and appropriately remove the squares on the various terms, in order to argue that $\sigma^*(R) = \tau_R + N_R$ is optimal. 

However, more elegantly, we can note that $\tau_R + N_R$ is a special case of the family in Eq.~\eqref{eqn_thm2_optimalset}, and focus on proving that the entire family is optimal.

Invoking Lemma~\ref{lemma3}, we can write:
\begin{align}
    2D_{T}(\rho,\sigma) &=  \sum_{i} |\lambda_{\sigma i} - \lambda_{ i} | \\
    &=  \sum_{i\in \SC} |\lambda_{\sigma i} - \lambda_{ i}| + \sum_{i\notin \SC} \lambda_{ i}\\
    &=  \sum_{i\in \SC} q_i + \sum_{i\notin \SC} \lambda_{ i}
\end{align}
where $\SC$ is the set of indices over which $\sigma$ has support, and $q_i = |\lambda_{\sigma i} - \lambda_{ i}|$. Let us define $\delta_i := \lambda_{\sigma i} - \lambda_{ i}$ and note that $\delta_i \leq q_i$. Therefore we obtain the bound:
\begin{align}
    2D_{T}(\rho,\sigma) &\geq  \sum_{i\in \SC} \delta_i + \sum_{i\notin \SC} \lambda_{ i}\\
    &= \sum_{i\in \SC}( \lambda_{\sigma i} - \lambda_{ i}) + \sum_{i\notin \SC} \lambda_{ i}\\
    &= 1 - \sum_{i\in \SC} \lambda_{ i} + \sum_{i\notin \SC} \lambda_{ i}\\
    &= 2 \sum_{i\notin \SC} \lambda_{ i}\,.\label{eq:TraceDistanceEquality}
\end{align}
It is clear that the bound $2D_{T}(\rho,\sigma)\geq 2 \sum_{i\notin \SC} \lambda_{ i}$ is saturated (i.e., becomes an equality) iff $\delta_i = q_i$ for all $i \in \SC$. In turn, this condition holds iff $\delta_i \geq 0$ for all $i \in \SC$, which is equivalent to the condition:
\begin{equation}\label{eqn_saturation_condition}
    \lambda_{\sigma i} \geq \lambda_{ i},\quad\text{for all }i \in \SC\,.
\end{equation}
Hence, for fixed choice of $\SC$, the minimum value of $D_{T}(\rho,\sigma) $ is $\sum_{i\notin \SC} \lambda_{ i}$, and this minimum is achieved iff Eq.~\eqref{eqn_saturation_condition} is satisfied.

Next we can consider varying the choice of $\SC$. Let $G(S) = \sum_{i\notin \SC} \lambda_{ i}$, which is the minimum value for $D_{T}(\rho,\sigma) $ for fixed $\SC$. It is clear that the minimum value of $G(S)$, over the choice of $\SC$, is achieved whenever the indices in $\SC$ correspond to the $R$ largest eigenvalues of $\rho$. In that case, $G(S) = \sum_{i >R} \lambda_{ i}$ is the sum over the $r-R$ smallest eigenvalues of $\rho$. Therefore, the optimal choice is $\SC = \{1,...,R\}$.
Therefore we have that the minimum value of $D_{T}(\rho,\sigma) $ is $\sum_{i>R} \lambda_{ i}$, and this minimum is achieved iff
\begin{equation}\label{eqn_theorem2_optimalstate}
    \sigma = \sum_{i=1}^R \lambda_{\sigma i}\ketbra{e_i}\,
\end{equation}
with
\begin{equation}\label{eqn_saturation_condition2}
    \lambda_{\sigma i} \geq \lambda_{ i},\quad\text{for }i = 1,...,R\,.
\end{equation}
The conditions in Eqs.~\eqref{eqn_theorem2_optimalstate} and \eqref{eqn_saturation_condition2} together define the optimal set, as stated in \eqref{eqn_thm2_optimalset}. This completes the proof.
\end{proof}

It follows from Eq.~\eqref{eq:TraceDistanceEquality} and $\SC = \{1,...,R\}$ that the optimal trace distance can be written as
\begin{equation}
    D^*_T(\rho, \sigma^*(R)) = 1 - \Tr[\Pi_R \rho] = \sum_{i=R+1}^r \lambda_i \, ,
\end{equation}
which is simply the sum of the eigenvalues of $\rho$ that are not approximated. This is clearly a monotonically decreasing function of $R$ and is only 0 when $R = r$. Further, it makes the meaning of the family of optimal states very clear: they are the states whose only error contribution is from the $r - R$ eigenvalues that cannot be approximated with a rank constraint $\rank(\sigma) \leq R$.

\subsection{\label{sec:comp-hs-tr}Comparison of Hilbert-Schmidt and Trace Distance}

The following example helps to illustrate our theorems above, especially how they differ for the Hilbert-Schmidt and trace distance.

Suppose we wish to approximate the state $\Tilde{\rho} = U \text{diag}(\vec{\Tilde{\lambda}}) U^{\dagger}$, $\vec{\Tilde{\lambda}} = (0.41, 0.39, 0.2, 0.0)$
with a rank 2 approximation. By Lemma~\ref{lemma1} and Lemma~\ref{lemma3}, the optimal approximation, denoted $\sigma^*(R = 2)$, must be diagonal in the same basis as $\Tilde{\rho}$. That is, $\sigma^* = U \text{diag}(\lambda^*_{\sigma 1}, \lambda^*_{\sigma 2}, 0, 0)U^{\dagger}$. By Theorem~\ref{theorem1}, the \emph{unique} optimal eigenvalues that minimize the Hilbert-Schmidt distance are $\lambda^*_{\sigma 1} = 0.51$ and $\lambda^*_{\sigma 2} = 0.49$ with corresponding optimal distance $D^*_{HS} \equiv D_{HS}(\Tilde{\rho}, \sigma^*(R)) = 0.06$. By Theorem~\ref{theorem2} this state also minimizes the trace distance $D^*_{T} \equiv D_T(\Tilde{\rho}, \sigma^*) = 0.4$, but it is not the only state that achieves $D^*_T$. This distinction between the unique solution for $D_{HS}$ and highly degenerate set of solutions for $D_T$ can be represented visually as shown in Fig.~\ref{fig:contour-plots}. 

\begin{figure}[ht]
    \centering
    \subfloat[The solution space of the QLRAP with the Hilbert-Schmidt distance is unique]{\label{fig:2norm-contour}\includegraphics[width=0.37\textwidth]{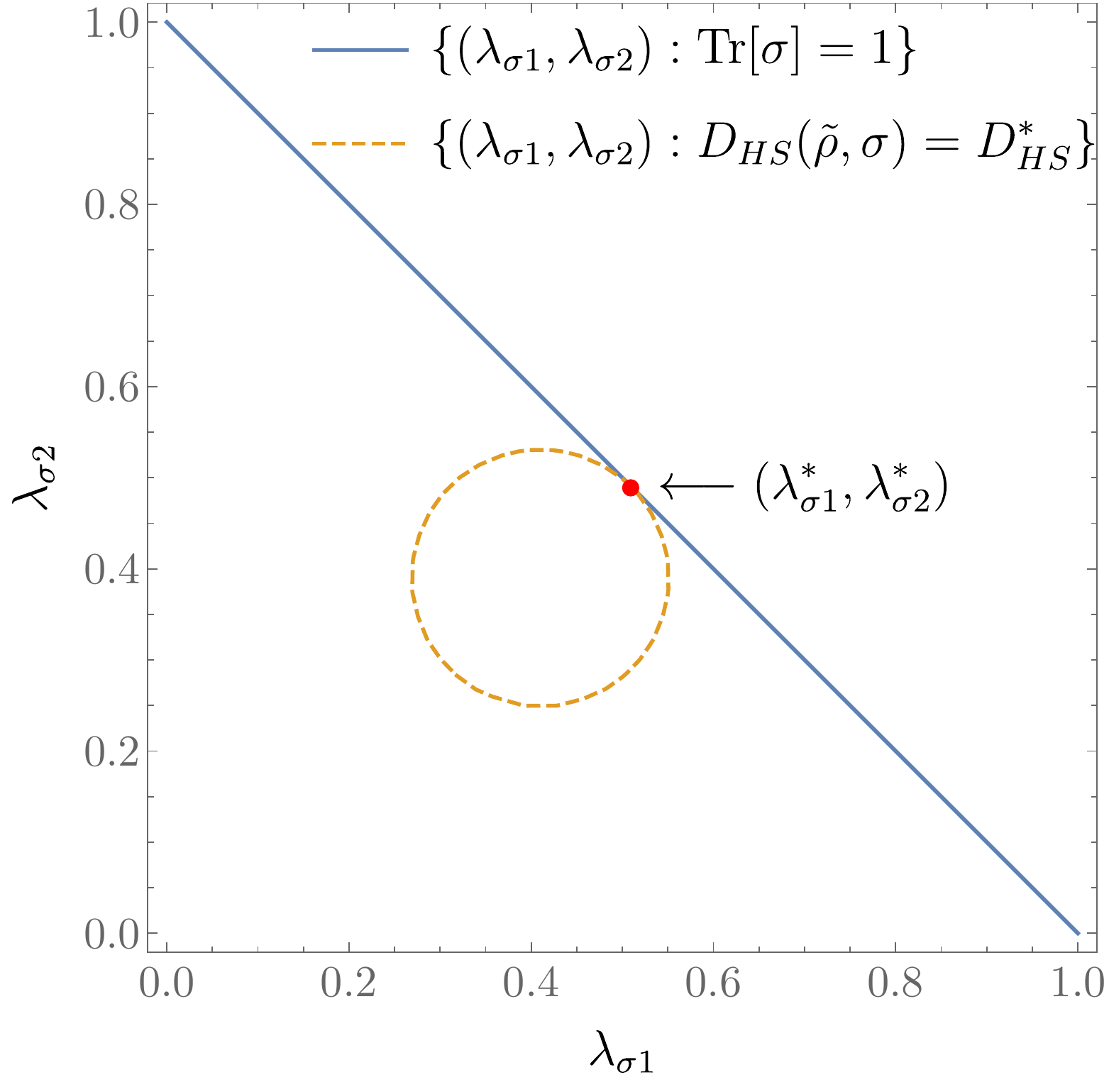}}
     \hfill
    \subfloat[The solution space of the QLRAP with the trace distance is highly degenerate]{\label{fig:1norm-contour}\includegraphics[width=0.37\textwidth]{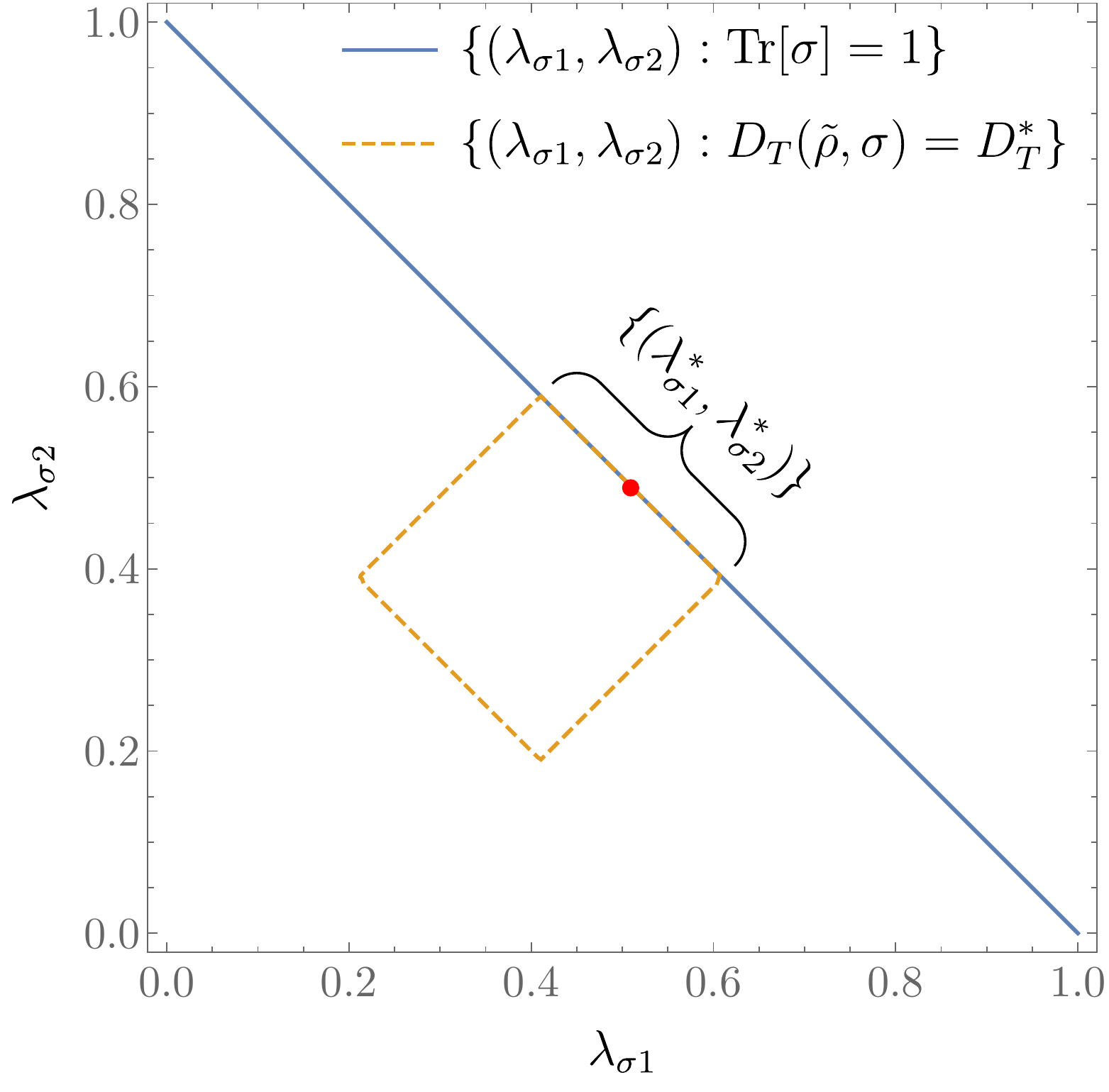}}
    \caption{\textbf{Comparison of QLRAP solutions for Hilbert-Schmidt distance and trace distance.} We explore the optimal rank 2 approximations of $\tilde{\rho}$ (defined in text) for the (a) Hilbert-Schmidt distance and the (b) trace distance. As explained in the text, the optimal set of solutions to QLRAP lie at the intersection of two constraints $\Tr[\sigma] = 1$ and $D(\Tilde{\rho}, \sigma) = D^*$. (We note that the point denoting $\Tilde{\rho}$ is not plotted as it sits outside of the plane shown). For $D_{HS}$, this set is unique and given by Theorem~\ref{theorem1}. For $D_T$, this set is highly degenerate and given by Theorem~\ref{theorem2}. In fact, two valid solutions for the trace distance can have different eigenvector orderings. }
    \label{fig:contour-plots}
\end{figure}

To explain the plot, we first assume that all states we consider are diagonal in the same basis. This means we may uniquely identify a state with its eigenvalues. Since $\rank(\Tilde{\rho}) = 3$, its eigenvalue vector is of the form $(\Tilde{\lambda}_{1}, \Tilde{\lambda}_{2}, \Tilde{\lambda}_{3}, 0)$ which we can visualize as a point in $\mathbb{R}^3$. A rank 2 approximation of $\rank(\Tilde{\rho})$ by definition only has two non-zero eigenvalues and so sits on the $\mathbb{R}^2$ slice through this space where $\lambda_{\sigma 3} = 0$. In this plane, the solution-space to the QLRAP is just the set of $\{(\lambda_{\sigma 1}, \lambda_{\sigma 2})\}$ that is normalized $\Tr[\sigma] = \lambda_{\sigma 1} + \lambda_{\sigma 2} = 1$ and achieves the optimal distance $D(\Tilde{\rho}, \sigma) = D^*$. For the Hilbert-Schmidt distance, the intersection of these two constraints is a single point as denoted by a labeled red dot in Fig.~\ref{fig:2norm-contour}. For the trace distance, the solution space is highly degenerate and consists of all points satisfying $\lambda_{\sigma 1} \geq 0.41$ and $\lambda_{\sigma 2} = 1 - \lambda_{\sigma 2} \geq 0.39$ and is denoted by a curly bracket in Fig.~\ref{fig:1norm-contour}.

That the trace distance leads to such a highly degenerate solution space is problematic in applications. For example, $\lambda^{'}_{\sigma} = (0.49, 0.51)$ is also a valid solution which swaps the ordering of the corresponding eigenvectors. Going beyond this example, the solutions $\sigma^*$ for the trace distance do not generally have a unique eigenvector ordering. The Hilbert-Schmidt distance, on the other hand, does obviously have a unique eigenvector ordering. Thus, in applications where the ordering of eigenvectors has an important meaning--such as PCA--the Hilbert-Schmidt distance should be preferred over the trace distance.

\section{\label{sec:PCA}Application: Principal Component Analysis}

We now discuss the application of our results to Principal Component Analysis (PCA). 

Quantum algorithms for PCA have been proposed, but often involve deep circuits and large qubit requirements~\cite{lloyd_quantum_2014}. More recently, some variational quantum algorithms for PCA have received attention, due to their resource efficiency~\cite{larose2019variational,cerezo2020variational,verdon2019quantum}. However, more research is needed to understand their scalability and performance on various datasets. 

Here we discuss how our results above could be used to construct a novel variational quantum algorithm for PCA. Suppose that one has a means to variationally prepare a mixed state $\sigma$ with a rank constraint, $\rank(\sigma)\leq R$. This would involve, say, acting with a parameterized quantum channel $\EC_{\vec{\theta}}$ to give \begin{equation}
\sigma_{\vec{\theta}} = \EC_{\vec{\theta}} (\ketbra{\vec{0}})
\end{equation}
Naturally there are various ways to implement $\EC_{\vec{\theta}}$. For example, one could implement a Stinespring dilation (i.e., prepare a purification of $\sigma_{\vec{\theta}}$) and then trace out the ancilla qubits, with the number of ancilla qubits controlling the rank of $\sigma_{\vec{\theta}}$. Alternatively, one could use classical randomness to randomly prepare the different eigenvectors of $\sigma_{\vec{\theta}}$, with the number of eigenvectors controlling the rank of $\sigma_{\vec{\theta}}$. 

Once $\sigma_{\vec{\theta}}$ is prepared on a quantum device, one can then efficiently estimate the Hilbert-Schmidt distance: 
\begin{equation}\label{eqn_dhs_cost}
D_{HS}(\rho,\sigma_{\vec{\theta}}) = \Tr(\rho^2) +  \Tr(\sigma_{\vec{\theta}}^2) - 2  \Tr(\rho \sigma_{\vec{\theta}}) 
\end{equation}
Each of the three terms in this expression can be efficiently estimated with the destructive SWAP test~\cite{cincio2018learning}, where one measures the SWAP operator to estimate the Hilbert-Schmidt inner product.  One can then use the distance in Eq.~\eqref{eqn_dhs_cost} as the cost function in a variational optimization loop. Minimizing this cost will result in learning the state $\sigma_{\vec{\theta}}^* = \tau_R +N_R $.

Depending how one prepares $\sigma_{\vec{\theta}}^*$, there are then various ways to extract the principal components. For example. if one prepared the purification of  $\sigma_{\vec{\theta}}^*$, then an appropriate measurement on the ancilla system would prepare the eigenvectors of $\sigma_{\vec{\theta}}^*$ on the other system. Alternatively, if one used classical randomness to prepare $\sigma_{\vec{\theta}}^*$, then one already has the quantum circuits to prepare its eigenvectors. Since the eigenvectors of $\sigma_{\vec{\theta}}^*$ correspond to the principal components of $\rho$, then this corresponds to performing PCA.

A key benefit of using the Hilbert-Schmidt distance is that the optimal state is unique, and also that the ordering of the eigenvalues of $\sigma_{\vec{\theta}}^*$ match the ordering of the eigenvalues of $\rho$. This means that one will infer the correct principal components of $\rho$ via $\sigma_{\vec{\theta}}^*$. In contrast, it is important to note that the degeneracy issues associated with trace distance make it less useful for PCA. Because the solution is non-unique for trace distance, there can exist some solutions that correspond to the wrong ordering of the eigenvalues. Hence, employing the trace distance could mislead one to infer the wrong eigenvalue ordering for $\rho$. We therefore advocate using the Hilbert-Schmidt distance, which also has the additional benefit of being efficiently estimatable on a quantum device.

\section{Conclusions}

In this work, we introduced the quantum low-rank approximation problem (QLRAP). We presented the complete solutions to the QLRAP for both the Hilbert-Schmidt distance and the trace distance. We found that the Hilbert-Schmidt distance yields a unique solution with the ordering of the eigenvalues matching those of the target state $\rho$. In contrast, the trace distance leads to a family of solutions, in which the eigenvalue ordering might not match that of the target state. We argued that this makes the Hilbert-Schmidt distance more useful for PCA applications than the trace distance.

In future work, we plan to further investigate using the  Hilbert-Schmidt distance for PCA applications, in the context of variational quantum algorithms. We expect that this will lead to a resource-efficient method for performing PCA on near-term quantum computers. We additionally plan to use our results to guide the development of algorithms to prepare low-rank approximations to mixed states. We expect such compression algorithms to find use on near-term noisy quantum hardware, where reducing resources required (e.g. number of qubits and/or number of circuits) to store and process quantum states is critical.

\begin{acknowledgments}

We thank Lukasz Cincio, Andrew Sornborger, and Mark Wilde for helpful conversations. We thank Daniel Lidar for suggested Proof 1 as a simple alternative to Proof 2 for Theorem 1. NE was supported by the U.S. Department of Energy (DOE) Computational Science Graduate Fellowship under Award Number DE-SC0020347. ZH acknowledges support from the Los Alamos National Laboratory (LANL) Mark Kac Fellowship. PJC acknowledges initial support from the LANL ASC Beyond Moore's Law project. PJC was also supported by the U.S. DOE, Office of Science, Office of Advanced Scientific Computing Research, under the Accelerated Research in Quantum Computing (ARQC) program.
\end{acknowledgments}

\bibliography{quantum.bib}

\end{document}